\documentclass[twocolumn,floatfix]{revtex4-1}
\UseRawInputEncoding
\usepackage{comment}
\usepackage{mathrsfs}
\usepackage{physics}
\usepackage{tikz}
\usepackage[caption=false]{subfig}
\usepackage[colorlinks=true,linkcolor=blue,citecolor=blue,urlcolor=blue]{hyperref}
\usepackage{amsmath,amssymb,amsfonts,amsthm}
\usepackage{soul}
\newtheorem*{thrm}{Theorem}

\usepackage{graphicx}
\usepackage{ragged2e}
\usepackage{float}
\usepackage{lipsum}

\newcommand{\kk}{\mathbf{k}}

\usepackage[normalem]{ulem}

\begin{document}
\date{\today}
\author{Nikolaos K. Kollas}
\email{kollas@upatras.gr}
\affiliation{Division of Theoretical and Mathematical Physics, Astronomy and Astrophysics, Department of Physics, University of Patras, 26504,  Patras, Greece}
\author{Dimitris Moustos}
\email{dmoustos@upatras.gr}
\affiliation{Division of Theoretical and Mathematical Physics, Astronomy and Astrophysics, Department of Physics, University of Patras, 26504,  Patras, Greece}
\title{Assisted harvesting and catalysis of coherence from scalar fields}
\begin{abstract}
    Recently it has been demonstrated that it is possible to harvest quantum resources other than entanglement from a coherent scalar field. Employing time-dependent perturbation theory, we present a complete analysis of the conditions under which a spatially extended Unruh-DeWitt detector coupled to the proper time derivative of the field can harvest coherence for any initial state of the field, as well as the energy cost that is required for
    each harvest. By studying harvesting under repeatable extractions it is proven that when the detector interacts with the field through a delta coupling coherence is catalytic. For a Gaussian smeared detector it is shown that harvesting from a coherent field depends on the phase of its amplitude distribution and its initial energy as well as on the mean radius of the detector and the mean interaction duration between the two. For a detector moving at a constant velocity and with a mean radius of the same order as its transition wavelength, we observe that, for relativistic speeds, coherence swelling effects are present the intensity of which depends on the dimension of the underlying Minkowski spacetime.
\end{abstract}
\maketitle
\section{Introduction}
Superposition is one of the most striking phenomena which distinguishes quantum from classical physics. The degree to which a system is superposed between different orthogonal states is known as \emph{coherence} \cite{RevModPhys.89.041003,aberg2006quantifying,PhysRevLett.119.230401}. Much like entanglement \cite{RevModPhys.81.865}, coherence is considered to be a valuable resource in quantum information processes. In Quantum computing \cite{nielsen2002quantum,Preskill2018quantumcomputingin}, where information is encoded in the states of two-level systems, algorithms designed to operate in superposition, are exponentially faster than their classical counterparts \cite{10.1137/S0097539795293172,PhysRevLett.79.325,Arute2019}. Coherence is so central to the development of a universal quantum computer that it is used as a metric for the quality of a quantum processor. The time that it takes for a qubit to effectively decohere due to noise is known as the \emph{dephasing time} with current processors achieving times of a few hundred microseconds. Coherent phenomena are important in other fields of research, such as quantum metrology \cite{PhysRevA.94.052324} and thermodynamics \cite{Lostaglio2015,PhysRevLett.115.210403,PhysRevX.5.021001,Narasimhachar2015,Korzekwa_2016} for example. Surprisingly it has been suggested that these phenomena might also be present in biological processes and more specifically in the efficiency of energy transport during photosynthesis \cite{Lloyd_2011}.

A simple method of obtaining coherence is by extracting it from another system. When this process involves a quantum field as the source then it is known as a \emph{coherence harvesting protocol}. Despite an extensive amount of research on entanglement harvesting protocols (see, e.g., \cite{VALENTINI1991321,Reznik2003,PhysRevA.71.042104,Salton_2015,PhysRevD.92.064042,PhysRevD.96.025020,PhysRevD.96.065008, PhysRevD.98.085007,Cong:2020nec,Tjoa:2020eqh}) and the deep connection that exists between entanglement and coherence \cite{PhysRevLett.115.020403,PhysRevLett.117.020402,PhysRevA.96.032316}, coherent harvesting has not received any attention. By employing the Unruh-DeWitt (UDW) particle detector model \cite{Unruh,DeWitt,birrell}, it was shown recently that a two-level pointlike detector, initially in its ground energy state, interacting with a coherent massless scalar field in $1+1$ flat spacetime, can harvest a small amount of coherence \cite{KBM}. As it turns out, this amount depends on the initial energy of the field, the mean interaction duration and the detector's state of motion. For a detector moving at relativistic speeds and initial field energies lower than the gap between its energy levels, it is possible to extract a larger amount of coherence than when it is static, a phenomenon dubbed by the authors as ``swelling''.  

In this article, we provide a thorough study of the conditions under which coherence harvesting is possible for any initial state of the field in $n+1$ dimensional Minkowski spacetime. In order to achieve this and to avoid the problem of IR divergences that are present in the $1+1$ dimensional case of a linear coupling between detector and field \cite{BJA}, we instead consider an interaction in which the former is coupled to the proper time derivative of the latter. Both models contain all the essential features of matter interacting with radiation \cite{Wavepacket:det,CHLI2}, so they provide a useful benchmark for studying possible applications of relativistic effects in quantum information processing. Acknowledging the fact that a pointlike detector is not a physical system--an atom or an elementary particle, for example, has finite size-- and to make our results as relevant as possible we will take into consideration the spatial extension of the detector. 

We show that when the interaction is instantaneous harvesting is catalytic \cite{PhysRevLett.113.150402}. At the cost of some energy, which assists in the extraction process, it is possible to repeatedly extract the same amount of coherence each time. For an inertial detector moving at a constant velocity and under suitable conditions, it is proven that this is also the maximum amount that can be obtained. As an example we consider the case of  harvesting coherence from a coherent scalar field and find that the process depends on the phase of its coherent amplitude distribution, its initial energy, the mean radius of the detector and the mean interaction duration between the two. For a mean radius comparable to the inverse of its transition frequency, it is shown that although the amount of coherence extracted is of the same order as the coupling constant the process can be repeated to obtain a single unit of coherence in a very short time. We conclude that even in the case of a spatially extended detector swelling effects are still present but these are weaker in a $3+1$ compared to a $1+1$ dimensional spacetime.

\section{Quantum coherence}\label{Qcoh}
From a physical point of view coherence reflects the degree of superposition that a quantum system exhibits when it simultaneously occupies different orthogonal eigenstates of an observable of interest \cite{PhysRevLett.119.230401}. Coherent systems are considered to be valuable resources in quantum information processes, because with their help it is possible, at the cost of consuming some of the coherence that they contain, to simulate transformations that violate conservation laws associated with the corresponding observable.

Mathematically, let $\{\ket{i}\}$ denote a set of basis states spanning a finite discrete Hilbert space $\mathcal{H}$, which correspond to the eigenstates of an observable $\hat{O}$. Any state $\rho$ which is diagonal in this basis
\begin{equation}\label{definition}
    \rho=\sum_ip_i\ketbra{i}
\end{equation}
is called \emph{incoherent} and commutes with the observable. If $\rho$ contains non-diagonal elements then it is called \emph{coherent} \cite{RevModPhys.89.041003}. In this case $[\rho,\hat{O}]\neq 0$ \footnote{assuming that the spectrum of $\hat{O}$ is non-degenerate}, and the state changes under the action of the one parameter group of symmetry transformations $U(s)=\text{exp}(-is\hat{O})$ generated by the observable. This makes coherent systems useful as reference frames and reservoirs for the implementation of non-symmetric transformations \cite{RevModPhys.79.555,Gour_2008,Marvian2014,PhysRevA.90.062110}. For example, for a fixed Hamiltonian $\hat{H}$, any system that possesses coherence with respect to the energy basis can be used as a clock since in this case its rate of change is non-zero, $\dot\rho(t)\neq0$, so it necessarily changes with the passage of time. The same system could alternatively be utilised as a coherent energy reservoir with the help of which it is possible to perform incoherent transformations on other systems \cite{PhysRevLett.113.150402}.

The amount of coherence present in a system can be quantified with the help of a \emph{coherence measure}. This is a real valued function $C(\cdot)$ on the set of density matrices $\mathcal{D}$ such that
\begin{equation}
    C(\rho)\geq 0,\quad \forall\rho\in\mathcal{D}
\end{equation}
with equality if and only if $\rho$ is incoherent. A simple example of such a function is given by the $\ell_1$-norm of coherence \cite{RevModPhys.89.041003}, which is equal to the sum of the modulus of the system's non-diagonal elements
\begin{equation}\label{coh_measure}
    C(\rho)=\sum_{i\neq j}\abs{\rho_{ij}}
\end{equation}
with values ranging between $0$ for an incoherent state and $d-1$ for the maximally coherent $d$-dimensional pure state
\begin{equation}
    \ket{\psi}=\frac{1}{\sqrt{d}}\sum_{i=0}^{d-1}\ket{i}.
\end{equation}

In order to extract coherence from a coherent system $\sigma$ to an incoherent system $\rho$ it is necessary to bring the two in contact and make them interact through a completely positive and trace preserving quantum operation. When the latter obeys the conservation law associated with the observable and is strictly incoherent (in the sense that it maps incoherent states to incoherent states) the process is called \emph{faithful} \cite{PhysRevA.101.042325}. When this is no longer the case the operation generates extra coherence, which increases the amount stored in the combined system and can assist in the extraction process \cite{PhysRevA.92.032331,BU20171670}, in much the same way that a quantum operation which is non-local can create entanglement between two spacelike separated systems. 

We shall now demonstrate how to construct such an assisted protocol for harvesting coherence onto an UDW detector from a scalar field. In what follows we shall assume a flat $n+1$ dimensional spacetime with metric signature $(-+\cdots+)$. We will denote spacetime vectors by sans-serif characters, and the scalar product of vectors $\mathsf{x}$ and $\mathsf{y}$ as $\mathsf{x}\cdot\mathsf{y}$. Boldface letters represent spatial n-vectors. Throughout, we make use of natural units in which $\hbar=c=1$ and employ the interaction picture for operators and states.
\section{Unruh-DeWitt detector model}\label{UDW}
To study the amount of coherence harvested from a massless scalar field we will employ an UDW detector coupled to the proper time derivative of the field \cite{Hinton,Takagi,DM}. In the simplest case considered here, the latter is modeled as a qubit with two energy levels, ground $\ket{g}$ and excited $\ket{e}$ and energy gap equal to $\Omega$, with Hamiltonian
\begin{equation}
    \hat{H}_\text{\tiny\ensuremath D}=\frac\Omega2(\ketbra{e}-\ketbra{g})
\end{equation}
which is moving along a worldline  $\mathsf{x}(\tau)$ parametrized by its proper time $\tau$. The detector is interacting with a massless  scalar field in $n+1$ dimensions 
\begin{equation}\label{field}
    \hat\phi(\mathsf{x})=\int\frac{d^n\kk}{\sqrt{(2\pi)^n2\abs{\kk}}}\left[\hat a_{\kk}e^{i\mathsf{k}\cdot\mathsf{x}}+\text{H.c.}\right],
\end{equation}
with a normal-ordered Hamiltonian of the form
\begin{equation}
    \hat{H}_\phi=\int \abs{\kk}\hat a^\dagger_{\kk}\hat a^{}_{\kk}d^n\kk,
\end{equation}
where $\hat a_{\kk}$, and $\hat a^\dagger_{\kk}$ are the creation and annihilation operators of the mode with momentum $\kk$ that satisfy the canonical commutation relations
\begin{equation}\label{commutation}
    [\hat a_{\kk},\hat a_{\mathbf{k'}}]=[\hat a^\dagger_{\kk},\hat a^\dagger_{\mathbf{k'}}]=0,\quad  [\hat a_{\kk},\hat a^\dagger_{\mathbf{k'}}]=\delta(\kk-\mathbf{k'}).
\end{equation}

The interaction between detector and field is constructed by coupling the former's monopole moment operator
\begin{equation}\label{dipole}
    \hat{\mu}(\tau)=e^{i\Omega\tau}\ketbra{e}{g}+ e^{-i\Omega\tau}\ketbra{g}{e},
\end{equation}
to the momentum degrees of freedom of the latter through the following interaction Hamiltonian
\begin{equation}\label{interaction}
    \hat{H}_{\text{int}}(\tau)=\lambda\chi(\tau)\hat{\mu}(\tau)\otimes\partial_\tau\hat{\phi}_f(\mathsf{x}(\tau)).
\end{equation}
 Here $\lambda$ is  a coupling constant with dimensions $(\mbox{length})^{\frac{n+1}{2}}$, $\chi(\tau)$ is a real valued \emph{switching function} that describes the way the interaction is switched on and off; and $\hat{\phi}_f(\mathsf{x}(\tau))$ is a smeared field on the detector's center of mass worldline  $\mathsf{x}(\tau)=(t(\tau),\mathbf{x}(\tau))$,
\begin{equation}\label{smeared}
    \hat{\phi}_f(\mathsf{x}(\tau))=\int_{\mathcal{S}(\tau)} f(\boldsymbol{\xi})\hat{\phi}(\mathsf{x}(\tau,\boldsymbol{\xi}))d^n\boldsymbol\xi,
\end{equation}
where 
\begin{equation}\label{Fermi-Walker}
    \mathsf{x}(\tau,\boldsymbol\xi)=\mathsf{x}(\tau)+\boldsymbol\xi
\end{equation}
are the Fermi-Walker coordinates \cite{MTW} on the simultaneity hyperplane $S(\tau)$, which is defined by all those space-like vectors $\boldsymbol\xi$ normal to the detector's four-velocity, $\mathcal{S}(\tau)=\left\{\boldsymbol\xi|\mathsf u\cdot\boldsymbol\xi=0\right\}$ (see Fig. \ref{fig:Fermi-Walker}).
\begin{figure}
    \centering
    \includegraphics[width=\columnwidth]{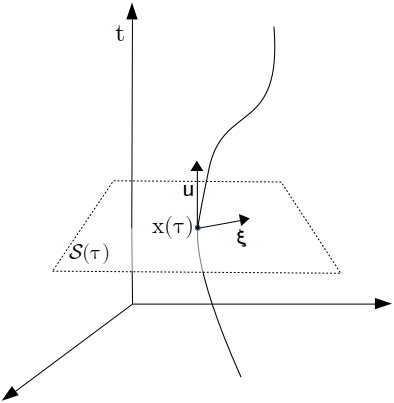}
    \caption{Any point in the neighbourhood of the detector's worldline can be described by its Fermi-Walker coordinates $(\tau,\boldsymbol\xi)$, where the proper time $\tau$ indicates its position along the trajectory and $\boldsymbol\xi$ is the displacement vector from this point lying on the simultaneity hyperplane consisting of all those space-like vectors normal to its four-velocity $\mathsf{u}$.}\label{fig:Fermi-Walker}
\end{figure} The real valued function $f(\boldsymbol\xi)$ in Eq. (\ref{smeared}) is known as the \emph{smearing function} and is a physical reflection of the finite size and shape of the detector \cite{Schlicht_2004,Louko_2006,Wavepacket:det,CHLI2}.

Compared to the usual UDW interaction in which the detector is linearly coupled to the field, the derivative coupling is free of the issue of IR divergences in the $1+1$ dimensional case which arise due to the massless nature of the field \cite{BJA}. The Hamiltonian in Eq. \eqref{interaction} resembles closely the dipole interaction between an atom with dipole moment $\mathbf{d}$ and an external electromagnetic field, since in this case the electric field operator is defined, in the Coulomb gauge, by means of the vector potential $\hat{\mathbf{A}}(t,\boldsymbol{x})$ as $\hat{\mathbf{E}}(t,\boldsymbol{x})=-\partial_t \hat{\mathbf{A}}(t,\boldsymbol{x})$ \cite{Scully}.

Combining Εq. (\ref{field}) with Εqs. (\ref{smeared})-(\ref{Fermi-Walker}) the smeared field operator reads
\begin{equation}\label{mod_smeared}
    \hat{\phi}_f(\mathsf{x}(\tau))=\int\frac{d^n\kk}{\sqrt{(2\pi)^n2\abs{\kk}}}\left[F(\mathsf k,\tau)\hat a_{\kk}e^{i\mathsf{k}\cdot\mathsf{x}(\tau)}+\text{H.c.}\right],
\end{equation}
where
\begin{equation}\label{Fourier}
    F(\mathsf k,\tau)=\int_{\mathcal{S}(\tau)} f(\boldsymbol\xi)e^{i\mathsf k\cdot\boldsymbol\xi}d^n\boldsymbol\xi
\end{equation}
is the Fourier transform of the smearing function. Now $\mathsf k$ can always be decomposed as
\begin{equation}
    \mathsf k=(\mathsf k\cdot\mathsf u)\mathsf u+(\mathsf k\cdot{\mathsf {\boldsymbol\zeta}}){\mathsf {\boldsymbol\zeta}}
\end{equation} 
for some unit vector ${\mathsf{\boldsymbol\zeta}}\in \mathcal{S}(\tau)$. Since for a massless scalar field $\mathsf k$ is light-like, it follows that $(\mathsf k\cdot \mathsf u)^2=(\mathsf k\cdot{\mathsf {\boldsymbol\zeta}})^2$. This means that for a spherically symmetric smearing function the Fourier transform in Eq. (\ref{Fourier}) is real and depends only on $\abs{\mathsf{k}\cdot\mathsf{u}}$,
\begin{equation}\label{const-smearing}
    F(\mathsf{k},\tau)=F(\abs{\mathsf{k}\cdot\mathsf{u}}).
\end{equation}
\section{Assisted harvesting and catalysis of quantum coherence}\label{sec:harvesting}
Suppose now that before the interaction is switched on at a time $\tau_\text{on}$, the combined system of detector and field starts out in a separable state of the form
\begin{equation}\label{separable}
    \ketbra{g}\otimes\sigma_\phi,
\end{equation}
where the detector occupies its lowest energy level and the field is in a state $\sigma_\phi$. The final state of the system after a time $\tau_\text{off}$ at which the interaction is switched off, can be obtained by evolving Eq. (\ref{separable}) with the unitary operator 
\begin{equation}\label{evolution}
    \hat{U}=\mathcal{T}\text{exp}\left(-i\int\limits_{\tau_\text{on}}^{\tau_\text{off}}\hat{H}_{\text{int}}(\tau)d\tau\right),
\end{equation}
where $\mathcal{T}$ denotes time ordering. Assuming that the switching function has a compact support we can extend the limits over $\pm\infty$. Setting 
\begin{equation}\label{Phi}
    \hat\Phi=\int\limits_{-\infty}^{+\infty}\chi(\tau)e^{-i\Omega\tau}\partial_\tau\hat\phi_f(\mathsf{x}(\tau))d\tau,
\end{equation}
Eq. (\ref{evolution}) can then be rewritten as
\begin{equation}\label{eq-evolution}
    \hat{U}=\text{exp}\left[-i\lambda(\ketbra{e}{g}\otimes\hat\Phi^\dagger+\ketbra{g}{e}\otimes\hat\Phi)\right].
\end{equation}
Tracing out the field degrees of freedom, one can obtain the state of the detector after the interaction which in this case is equal to
\begin{equation}\label{det-state}
    \rho_\text{\tiny\ensuremath D}=\left(\begin{array}{cc}
         1-\lambda^2\tr^{}(\hat\Phi^\dagger\sigma_\phi\hat\Phi)& i\lambda \tr^{}(\hat\Phi\sigma_\phi) \\
         -i\lambda \tr^{}(\hat\Phi^\dagger\sigma_\phi)& \lambda^2\tr^{}(\hat\Phi^\dagger\sigma_\phi\hat\Phi)
    \end{array}\right)+\mathcal{O}(\lambda^3).
\end{equation}
In a similar fashion, by taking the partial trace over the detector's Hilbert space, we can obtain the state of the field after harvesting,
\begin{equation}\label{field-state}
    \sigma_\phi'=\sigma_\phi+\lambda^2\left(\hat\Phi^\dagger\sigma_\phi\hat\Phi-\frac12\left\{\hat\Phi\hat\Phi^\dagger,\sigma_\phi\right\}\right)+\mathcal{O}(\lambda^4).
\end{equation}

With the help of Eqs. (\ref{coh_measure}) and (\ref{det-state}) the amount of coherence harvested to the detector to lowest order in the coupling constant is equal to
\begin{equation}\label{coher}
    C=2\lambda\abs{\tr^{}(\hat\Phi\sigma_\phi)}.
\end{equation}
Defining
\begin{equation}
    \mathcal{F_\pm}(\kk)=\int\limits_{-\infty}^{+\infty}\chi(\tau)e^{\pm i\Omega\tau}\partial_\tau \left(F(\mathsf k,\tau)e^{i\mathsf k\cdot\mathsf x(\tau)}\right)d\tau,
\end{equation}
Eq. (\ref{coher}) can be written as
\begin{equation}\label{explicit-coher}
    C=2\lambda\abs{\int\frac{d^n\kk}{\sqrt{(2\pi)^n2\abs{\kk}}}\Big(\mathcal{F}_-(\kk)a(\kk)+\mathcal{F}^*_+(\kk)a^*(\kk)\Big)},
\end{equation}
where 
\begin{equation}\label{ampl-distr}
    a(\kk)=\tr^{}(\hat a_{\kk}\sigma_\phi)
\end{equation} is the \emph{coherent amplitude distribution} of the field.

Suppose that we wish to repeat the process and extract coherence onto a fresh detector copy. It is straightforward to see that for the $m$-th harvest one can extract an amount of
\begin{equation}\label{m-coher}
    C^{(m)}=2\lambda\abs{\tr^{}(\hat\Phi\sigma_\phi^{(m)})}
\end{equation}
units of coherence from a perturbed field in the state
\begin{equation}\label{m-field-state}
    \sigma_\phi^{(m)}=\sigma_\phi^{(m-1)}+\lambda^2\left(\hat\Phi^\dagger\sigma_\phi^{(m-1)}\hat\Phi-\frac12\left\{\hat\Phi\hat\Phi^\dagger,\sigma_\phi^{(m-1)}\right\}\right).
\end{equation}
By combining Eqs. (\ref{m-coher}) and (\ref{m-field-state}) and exploiting the cyclic property of the trace as well as the fact that $[\hat\Phi,\hat\Phi^\dagger]$ is a $c$-number (for proof see Appendix \ref{appendix-useful}) it follows that 
\begin{equation}\label{repeat-coh}
    C^{(m+1)}=C^{(m)}\abs{1+\frac{\lambda^2}{2}\left[\hat\Phi,\hat\Phi^\dagger\right]},
\end{equation}
so to lowest order in the coupling constant the amount of coherence harvested each time remains the same.
\begin{figure}
    \centering
    \includegraphics[width=\columnwidth]{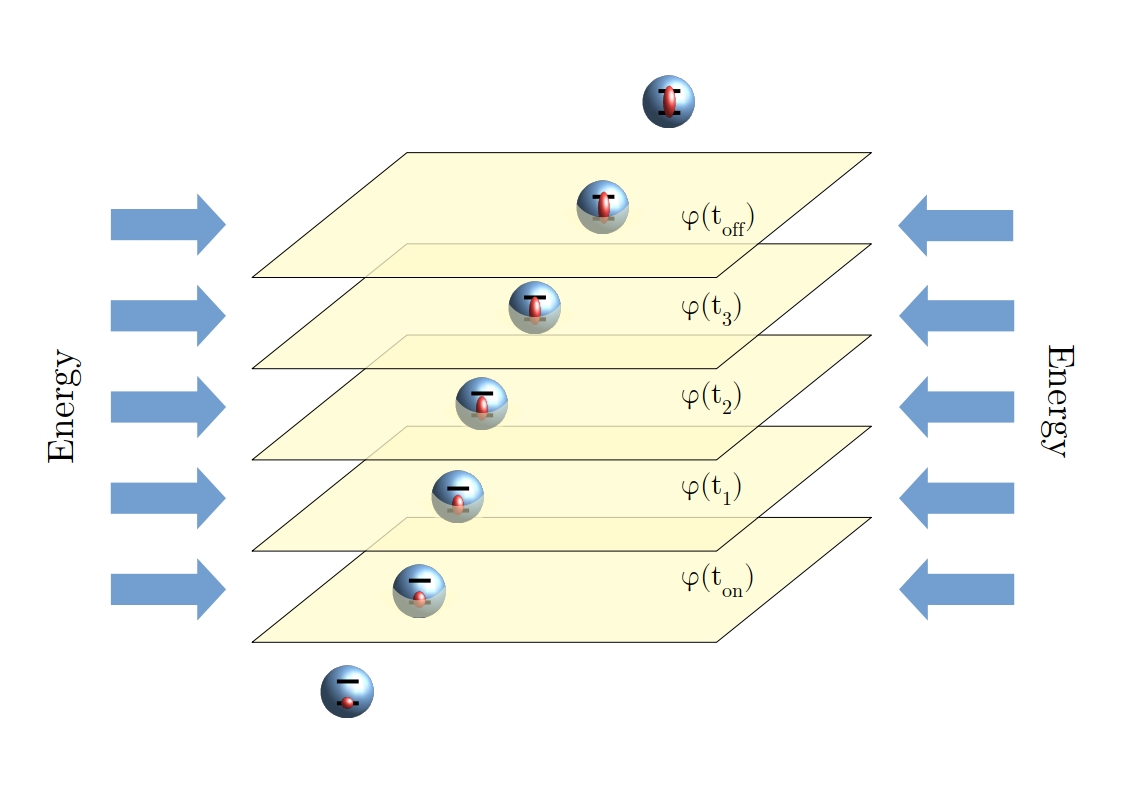}
    \caption{\textbf{Assisted harvesting of quantum coherence.} A moving two-level system, initially in its ground state at some time $t<t_\text{on}$, interacts with a massless scalar field through a derivative coupling. The process requires an external flow of  energy which assists harvesting by increasing the combined system's coherence. After the interaction is switched off at a time $t_\text{off}$ the detector will find itself in a superposition between its energy levels.}
    \label{fig:harvesting}
\end{figure}

Let's focus our attention on normalised smearing and switching functions such that
\begin{equation}\label{fixing}
    \int\limits_{-\infty}^{+\infty}\chi(\tau)d\tau=\int_{\mathcal S(\tau)}f(\boldsymbol\xi)d^n\boldsymbol\xi=1,
\end{equation}
and define
\begin{equation}
    R=\int_{\mathcal S(\tau)}\abs{\xi}f(\boldsymbol\xi)d^n\boldsymbol\xi
\end{equation}
as the mean radius of the detector and
\begin{equation}
    T=\int\limits_{-\infty}^{+\infty}\abs{\tau}\chi(\tau)d\tau
\end{equation}
as the mean interaction duration respectively. This will make it easier to compare different setups and will allow the study, in a unified way, of the effects that different sizes and finite interaction durations have on harvesting as well as the limiting case of an instantaneous interaction in which $\chi(\tau)=\delta(\tau)$. In this limit, $[\hat\Phi,\hat\Phi^\dagger]=0$ and the amount harvested each time is exactly the same to any order (for more details see Appendix \ref{second-appendix}). It seems that when the detector interacts with the field through a delta coupling, coherence harvesting is catalytic \cite{PhysRevLett.113.150402}. Even though in principle this is allowed for infinite dimensional systems that act as coherence reservoirs \cite{PhysRevLett.123.020403,PhysRevLett.123.020404}, it is not certain if this is the case here. Since the interaction Hamiltonian does not commute with the unperturbed part, $\hat H_\text{\tiny\ensuremath D}+\hat H_\phi$, of the total Hamiltonian, the process requires an outside supply of positive energy $\Delta E$ each time \cite{Hackl2019minimalenergycostof,Bény2018}. Energy non-conserving unitaries like the one in Eq. (\ref{evolution}) can increase the coherence of the combined system assisting in the extraction process \cite{PhysRevA.92.032331,BU20171670} (see Fig. \ref{fig:harvesting}). Nonetheless a necessary condition for extracting a non trivial amount of coherence is for the field to be in a state with a non-zero coherent amplitude distribution.
\section{Inertial detectors}\label{Sec:Inert}
We will now consider an inertial detector which is moving along a worldline  with a constant velocity $\boldsymbol\upsilon$, and whose center of mass coordinates is given by
\begin{equation}\label{coordinates}
    \mathsf{x}(\tau)=\mathsf{u}\tau,
\end{equation}
where $\mathsf{u}=\gamma(1,\boldsymbol{\upsilon})$ is its four-velocity, with $\gamma=1/\sqrt{1-\upsilon^2}$ the Lorentz factor. For a spherically symmetric smearing function with a positive Fourier transform, it can be proven that
\begin{thrm}
For a suitable choice of the coherent amplitude distribution's phase the maximum amount of harvested coherence to lowest order, is obtained by a detector interacting instantaneously with the field.
\end{thrm}
\begin{proof}
Taking the absolute value inside the integral in Eq. (\ref{explicit-coher}) we find that
\begin{equation}
    C\leq2\lambda{\int\frac{d^n\kk}{\sqrt{(2\pi)^n2\abs{\kk}}}\abs{a(\kk)}\left(\abs{\mathcal{F}_-(\kk)}+\abs{\mathcal{F}_+(\kk)}\right)}.
\end{equation}
For a detector moving with a constant velocity the Fourier transform of the smearing function no longer depends on its proper time, in this case
\begin{equation}
    \mathcal{F}_-(\kk)=i(\mathsf k\cdot\mathsf u)F(\abs{\mathsf k\cdot\mathsf u})X^*(\Omega-\mathsf k\cdot\mathsf u)
\end{equation}
and
\begin{equation}
    \mathcal{F}_+(\kk)=i(\mathsf k\cdot\mathsf u)F(\abs{\mathsf k\cdot\mathsf u})X(\Omega+\mathsf k\cdot\mathsf u)
\end{equation}
where
\begin{equation}\label{Xi}
    X(\Omega\pm\mathsf k\cdot\mathsf u)=\int\limits_{-\infty}^{+\infty}\chi(\tau)e^{i(\Omega\pm\mathsf k\cdot\mathsf u)\tau}d\tau.
\end{equation}
Because of the normalization property in Eq. (\ref{fixing}), $\abs{X(\Omega\pm\mathsf k\cdot\mathsf u)}\leq 1$ so finally
\begin{equation}\label{inequality}
    C\leq4\lambda\int\frac{(-\mathsf k\cdot\mathsf u)}{\sqrt{(2\pi)^n2\abs{\kk}}}F(\abs{\mathsf k\cdot\mathsf u})\abs{a(\kk)}d^n\kk,
\end{equation}
where equality holds for $\chi(\tau)=\delta(\tau)$ and a coherent amplitude distribution with phase, $\text{arg}(a(\kk))=\frac{\pi}{2}$ \footnote{In the Unruh-DeWitt interaction where the factor $(-\mathsf k\cdot\mathsf u)$ in the numerator is absent, the Theorem holds for an arbitrary motion of the detector as long as $\mathsf{x}(0)=0$.}.

Note that if the Fourier transform of the smearing function is not positive then Eq. (\ref{inequality}) is only an upper bound on the amount of harvested coherence.
\end{proof}

If the amplitude distribution is also spherically symmetric then
\begin{multline}\label{moving-coher}
   C=2\lambda\left|\int\frac{(-\mathsf k\cdot\mathsf u)F(\abs{\mathsf k\cdot\mathsf u})}{\sqrt{(2\pi)^n2\abs{\kk}}}\left[a(\abs{\kk})X^*(\Omega-\mathsf k\cdot\mathsf u)\right.\right.\\\left.-{a}^*(\abs{\kk})X(\Omega+\mathsf k\cdot\mathsf u)\right]d^n\kk\Bigg|,  
\end{multline}
which for a static detector reduces to
\begin{multline}\label{static-coher}
   C=\frac{2\lambda s_n}{\sqrt{2(2\pi)^n}}\left|\int_0^{\infty}k^{n-\frac12}F(k)\left[a(k)X^*(\Omega+k)\right.\right.\\\left.-a^*(k)X(\Omega-k)\right]dk\Bigg|, \end{multline}
where $s_n=\frac{2\pi^{n/2}}{\Gamma(n/2)}$ is the surface area of the unit $n$-sphere. 
By boosting the four-momentum $\mathsf k$ to the detector's frame of reference it can be shown that Eq. (\ref{moving-coher}) is equivalent to Eq. (\ref{static-coher}) with a symmetric coherent amplitude distribution of the form
\begin{equation}\label{moving-distribution}
    a_{\upsilon}(k)=\frac1{s_n}\int a\left(\frac{k}{\gamma(1-\mathbf{\boldsymbol{\upsilon}}\cdot\hat\kk)}\right)\frac{d\hat\kk}{{[\gamma(1-\mathbf{\boldsymbol{\upsilon}}\cdot\hat\kk)]^{n-\frac12}}}.
\end{equation}
From the detector's point of view, the field's coherent amplitude is equivalent to a mixture of Doppler shifted distributions with weight equal to $[s_n\gamma(1-\boldsymbol{\upsilon}\cdot\hat\kk)^{n-\frac12}]^{-1}$. For a similar result regarding the interaction of an inertial detector with a heat bath see \cite{PhysRevD.102.085005}.
\section{Assisted harvesting and catalysis from a coherent field}\label{sec:harvest}
For a coherent state $\ket{a}$ of the field, the coherent amplitude distribution in Eq. (\ref{ampl-distr}) is equal to the eigenvalue of the annihilation operator with mode $\kk$
\begin{equation}
    \hat{a}_{\kk}\ket{a}=a(\kk)\ket{a},
\end{equation}
in this case the amount of harvested coherence to lowest order is given by the expectation value of the field operator $\hat\Phi$
\begin{equation}
    C=2\lambda|{\bra{a}\hat{\Phi}\ket{a}}|.
\end{equation}

The energetic cost associated with harvesting is equal to the energy difference between the final and initial states of the combined system of detector and field
\begin{equation}
    \Delta E=\tr(\hat H_\text{\tiny\ensuremath D}(\rho_\text{\tiny\ensuremath D}-\ketbra{g}))+\tr(\hat H_\phi(\sigma_\phi'-\ketbra{a})).
\end{equation}
To lowest order this splits into two contributions
\begin{equation}
    \Delta E=\Delta E_\text{coh}+\Delta E_\text{vac},
\end{equation}
where
\begin{equation}
    \Delta E_\text{coh}=\frac{C^2}{4}\left(\Omega+4\Re\left[\frac{\bra{a}[\hat\Phi,\hat{H}_\phi]\ket{a}}{\bra{a}\hat\Phi\ket{a}}\right]\right)
\end{equation}
is the cost associated with harvesting and
\begin{equation}
    \Delta E_\text{vac}=\frac{\lambda^2}{2(2\pi)^n}\int\left(1+\frac{\Omega}{\abs{\kk}}\right)\abs{\mathcal{F}_-(\kk)}^2d^n\kk.
\end{equation}
is the cost of interacting with the vacuum \cite{PhysRevD.96.025020}.

Let us consider an inertial detector and a harvesting process in which the switching and smearing functions are respectively given by the following Gaussians 
\begin{equation}\label{gaussian-switching}
    \chi(\tau)=\frac{\text{exp}\left(-\frac{\tau^2}{\pi T^2}\right)}{\pi T}
\end{equation}
\begin{equation}\label{gaussian-smearing}
    f(\boldsymbol\xi)=\frac{\text{exp}\left(-\frac{\boldsymbol\xi^2}{\pi R_n^2}\right)}{(\pi R_n)^n},
\end{equation}
while the state of the field is described by a coherent amplitude distribution with a unit average number of excited quanta of the form
\begin{equation}\label{amplitude}
    a(\kk)=\frac{\text{exp}(-\frac{k^2}{2\pi {E}_n^2}+i\frac{\pi r}{2})}{(\pi {E}_n)^{n/2}},\quad r=0,1
\end{equation}
where
\begin{equation}
    E_n=\frac{s_{n+1}}{\pi s_n}E\quad\text{and}\quad R_n=\frac{s_{n+1}}{\pi s_n}R,
\end{equation}
with $E=\bra{a}\hat H_\phi\ket{a}$ the mean initial energy of the field. Note that even though the support of Eq. (\ref{gaussian-switching}) is no longer compact, as was originally required, the analysis is expected to  present a good approximation to a compact switching function of the form
\begin{equation}
    \chi(\tau)=\begin{cases} 
    {\text{exp}(-\frac{\tau^2}{\pi T^2})}/(\pi T),&|\tau|\leq\mathcal{T}\\
    0,&\text{otherwise}
    \end{cases}
\end{equation}
provided that $\mathcal{T}\geq 4\sqrt{\pi}T$. We will now treat the static and moving cases separately.
\subsection{Static detector}
For $\upsilon=0$ the Fourier transforms of the switching and smearing functions are equal to
\begin{equation}\label{tint}
        X(\Omega\pm k)=\text{exp}\left[{-\frac{\pi(\Omega\pm k)^2T^2}{4}}\right]
\end{equation}
and
\begin{equation}\label{calculated-Fourier}
    F(\mathsf{k})=\text{exp}\left[-\frac{\pi k^2R_n^2}{4}\right]
\end{equation}
respectively. Inserting these into Eq. (\ref{static-coher}) we obtain that the amount of harvested coherence, which now depends on the initial energy of the field, the mean interaction duration and the mean radius of the detector is
\begin{multline}\label{C}
   C(E,T,R)=\frac{4\lambda s_n}{\sqrt{2(2\pi^2{E}_n)^n}}e^{-\frac{\pi\Omega^2T^2}{4}}\\
\times\int_0^\infty k^{n-\frac12}e^{-\mathrm{a}{k^2}}\sinh^{1-r}(bk)\cosh^r(bk)dk,
\end{multline}
with
\begin{equation}
    \mathrm{a}=\frac{1}{2\pi {E}_n^2}\left[1+\frac{\pi^2E_n^2(R^2_n+T^2)}{2}\right], \quad b=\frac{\pi \Omega T^2}{2}.
\end{equation}
The integral on the right hand side is equal to
\begin{figure*}
\begin{minipage}{\textwidth}
\subfloat[$r=1$]{\includegraphics[width=0.45\textwidth]{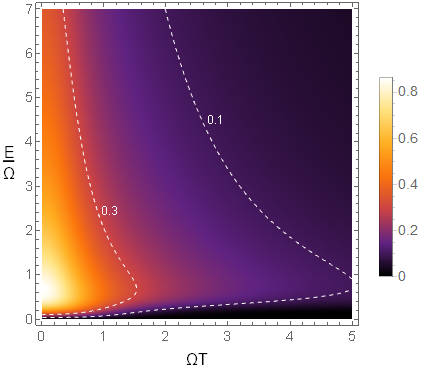}}\hspace{0.5cm}
\subfloat[$r=0$]{\includegraphics[width=0.45\textwidth]{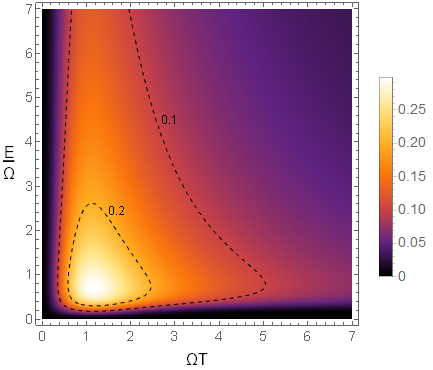}}
\end{minipage}
\caption{Amount of harvested coherence $C/\bar\lambda$ from a coherent scalar field in $1+1$ dimensions and a Gaussian amplitude distribution with phase a) $\phi=\frac{\pi}{2}$ and b) $\phi=0$, as a function of the mean initial energy of the field (in units $\Omega$) and the mean interaction duration (in units $1/\Omega$), for a detector with mean radius $R=1/\Omega$.}
\label{fig:static1}
\begin{minipage}{\textwidth}
\subfloat[$r=1$]{\includegraphics[width=0.45\textwidth]{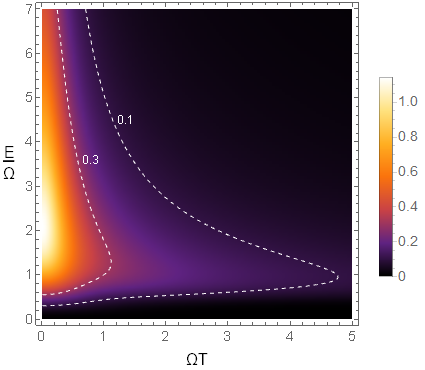}}\hspace{0.5cm}
\subfloat[$r=0$]{\includegraphics[width=0.45\textwidth]{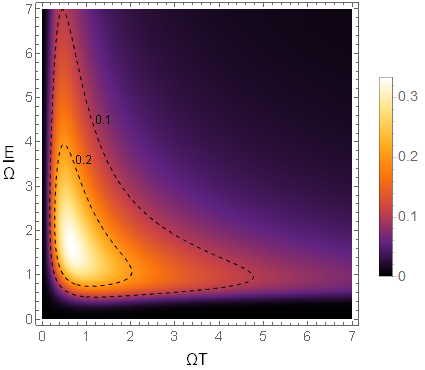}}
\end{minipage}
\caption{Amount of harvested coherence $C/\bar\lambda$ from a coherent scalar field in $3+1$ dimensions and a Gaussian amplitude distribution with phase a) $\phi=\frac{\pi}{2}$ and b) $\phi=0$, as a function of the mean initial energy of the field (in units $\Omega$) and the mean interaction duration (in units $1/\Omega$), for a detector with mean radius $R=1/\Omega$}\label{fig:static3}
\end{figure*}
\begin{widetext}
\begin{equation}\label{cylindrical}
    \int\limits_{0}^\infty k^{n-\frac{1}{2}}e^{-\mathrm{a}k^2}\sinh^{1-r}(bk)\cosh^r(bk)dk=
\frac{\Gamma(n+\frac12)}{2(2\mathrm{a})^{\frac{n}2+\frac14}}e^\frac{b^2}{8\mathrm{a}}\left[D_{-n-\frac12}\left(-\frac{b}{\sqrt{2\mathrm{a}}}\right)-(-1)^{r}D_{-n-\frac12}\left(\frac{b}{\sqrt{2\mathrm{a}}}\right)\right],\quad b>0
\end{equation}
where $D_p(z)$ denotes the \emph{parabolic cylinder function} \cite{gradshteyn2014table}. In a similar way it can be shown that
\begin{equation}
    \Delta E_\text{coh}=\frac{C^2}{4}\left[\Omega-\frac{4(n+\frac12)}{\sqrt{2\mathrm{a}}}\frac{D_{-n-\frac32}\left(-\frac{b}{\sqrt{2\mathrm{a}}}\right)+(-1)^{r}D_{-n-\frac32}\left(\frac{b}{\sqrt{2\mathrm{a}}}\right)}{D_{-n-\frac12}\left(-\frac{b}{\sqrt{2\mathrm{a}}}\right)-(-1)^{r}D_{-n-\frac12}\left(\frac{b}{\sqrt{2\mathrm{a}}}\right)}\right]
\end{equation}
and
\begin{equation}
    \Delta E_\text{vacuum}=\frac{\lambda^2\pi s_n\Gamma(n+1)}{(8\pi^2\mathrm{a}')^\frac{n+1}{2}}e^{-\frac{\pi\Omega^2T^2}{2}+\frac{b^2}{8\mathrm{a}'}}\left[\frac{n+1}{\sqrt{2\mathrm{a}'}}D_{-n-2}\left(\frac{2b}{\sqrt{2\mathrm{a}'}}\right)+\Omega D_{-n-1}\left(\frac{2b}{\sqrt{2\mathrm{a}'}}\right)\right],
\end{equation}
where 
\begin{equation}
    \mathrm{a}'=\frac{\pi(R_n^2+T^2)}{2}.
\end{equation}
\end{widetext}

In Figs. \ref{fig:static1} and \ref{fig:static3} we present the amount of coherence harvested, scaled by the dimensionless coupling constant $\bar\lambda=\lambda\Omega^\frac{n+1}{2}$, as a function of the initial mean energy $E$ of the field (in units $\Omega$) and the interaction duration $T$ (in units $1/\Omega$) for a $1+1$ and a $3+1$ dimensional Mikowski spacetime respectively. In order to simplify the situation we will tacitly assume from now on that the mean radius of the qubit is equal to its transition wavelength $R=1/\Omega$. It is clear from both figures that the harvesting profile depends strongly on the phase of the coherent amplitude distribution. For $r=1$ and for a fixed initial field energy, the maximum amount that can be harvested is obtained through the use of an instantaneous interaction ($T=0$), in agreement with the Theorem of Sec. \ref{Sec:Inert}. When $r=0$ it is impossible to harvest coherence to a qubit interacting instantaneously with the field, in this case the maximum is obtained for interaction durations comparable to the mean radius. In both settings, if the initial energy of the field is zero the amount of coherence harvested vanishes. This is also true for very large energy values. Qualitatively, harvesting is more efficient for field energies comparable to the energy gap. For a resonant energy of the field, $E=\Omega$, it is possible to extend the process to greater interaction times compared to other energies and still extract a small amount of coherence.

Now  with the help of Eq. (\ref{commutator}) of Appendix \ref{appendix-useful}, Eqs.  (\ref{tint})-(\ref{calculated-Fourier}) and Eq. (\ref{cylindrical}) it can be shown that
\begin{multline}
    \lambda^2[\hat\Phi,\hat\Phi^\dagger]=-\frac{2ns_n\bar\lambda^2}{s_{2n}[4\pi\Omega^2(R_n^2+T^2)]^\frac{n+1}{2}}e^{-\frac{\pi\Omega^2T^2(2R^2_n+T^2)}{4(R^2_n+T^2)}}\\\times\left[D_{-n-1}\left(\!-\sqrt\frac{\pi\Omega^2 T^4}{R^2_n+T^2}\right)\!-D_{-n-1}\left(\sqrt\frac{\pi\Omega^2T^4}{R^2_n+T^2}\right)\right].
\end{multline}
From Fig. \ref{fig:com} it can be seen that for $\bar\lambda<<1$ and $R=1/\Omega$ this term is negligible. Since the maximum amount of harvested coherence is of the same order as $\bar\lambda$ then, according to Eq. (\ref{repeat-coh}), we can repeat the process $m$ times for a total of $C_{\text{tot}}=\mathcal{O}(m\bar\lambda)$ units of coherence. Assuming that for a phase-less coherent amplitude distribution obtaining the maximum in each harvest requires a time of approximately $T=1/\Omega$ it follows the total duration is of the order $\mathcal{O}(m/\Omega)$. To extract a single unit of coherence requires therefore approximately $\mathcal{O}(1/\bar\lambda\Omega)$ seconds. For a transition frequency in the optical spectrum and $\bar\lambda=10^{-3}$ this time is of the order of $10^{-12}$ seconds.
\begin{figure}
    \includegraphics[width=\columnwidth]{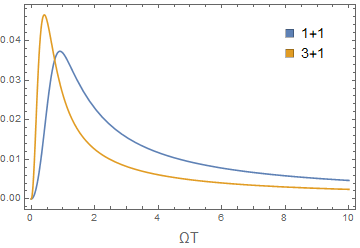}
    \caption{$\lambda^2[\hat\Phi,\hat\Phi^\dagger]/\bar\lambda^2$ as a function of the mean interaction duration (in units $1/\Omega$), for a detector with mean radius $R=1/\Omega$.}
    \label{fig:com}
\end{figure}
\subsection{Detector moving at a constant velocity}
\begin{figure*}\begin{minipage}{\textwidth}
\subfloat[$C_{0.8}/\bar\lambda$ ($r=1$)]{\includegraphics[scale=0.36]{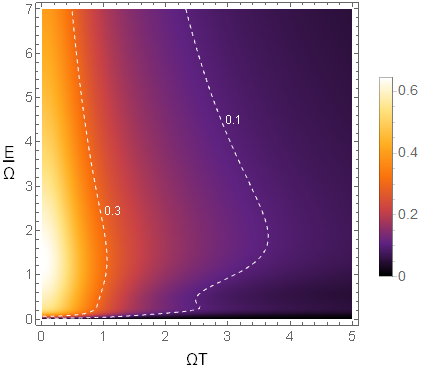}}\hspace{0.45cm}
\subfloat[$C_0/C_{0.8}$ ($r=1$)]{\includegraphics[scale=0.36]{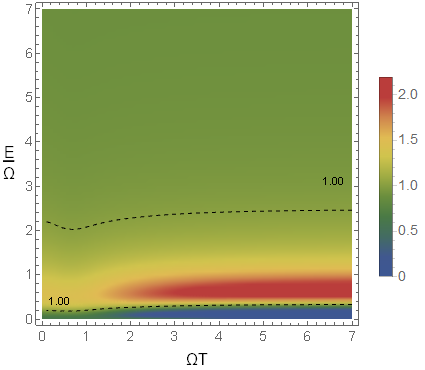}}\hspace{0.45cm}
\subfloat[$E=0.1\Omega$ ($r=1$)]{\includegraphics[scale=0.38]{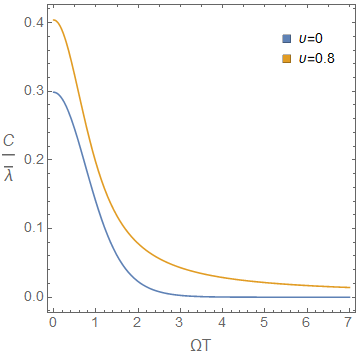}}\\
\subfloat[$C_{0.8}/\bar\lambda$ ($r=0$)]{\includegraphics[scale=0.36]{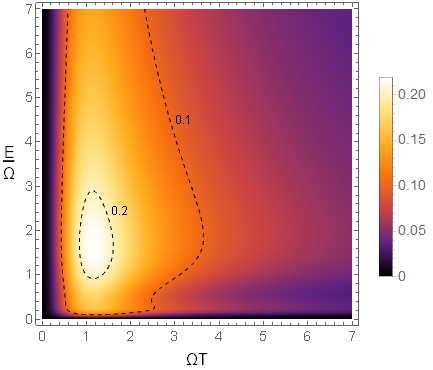}}\hspace{0.42cm}
\subfloat[$C_0/C_{0.8}$ ($r=0$)]{\includegraphics[scale=0.36]{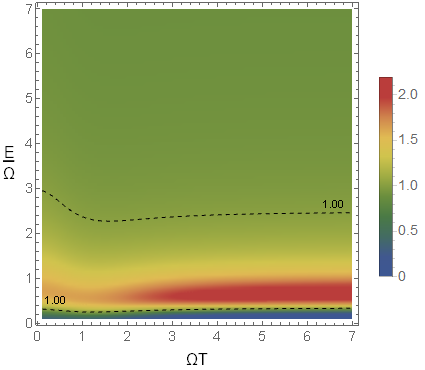}}\hspace{0.45cm}
\subfloat[$E=0.1\Omega$ ($r=0$)]{\includegraphics[scale=0.38]{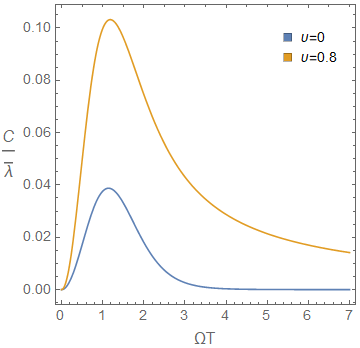}}
\end{minipage}
\caption{Left: Amount of harvested coherence, $C_{0.8}/\bar\lambda$, in $1+1$ dimensions. Center: Amount of swelling $C_0/C_{0.8}$. Right: Comparison between a static and a moving detector for an initial energy of the field $E=0.1\Omega$.}
\label{fig:6}
\begin{minipage}{\textwidth}
\subfloat[$C_{0.8}/\bar\lambda$ ($r=1$)]{\includegraphics[scale=0.36]{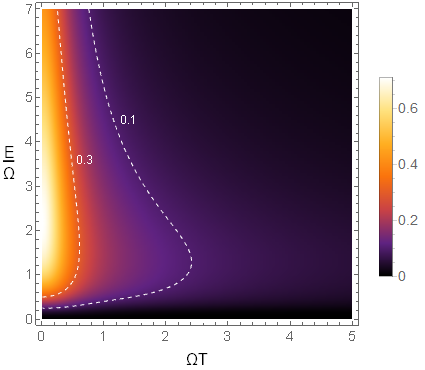}\hspace{0.5cm}}
\subfloat[$C_0/C_{0.8}$ ($r=1$)]{\includegraphics[scale=0.36]{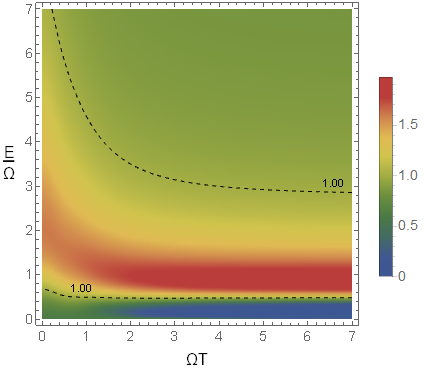}}\hspace{0.45cm}
\subfloat[$E=0.2\Omega$ ($r=1$)]{\includegraphics[scale=0.38]{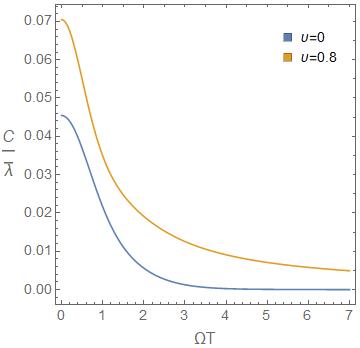}}\\
\subfloat[$C_{0.8}/\bar\lambda$ ($r=0$)]{\includegraphics[scale=0.36]{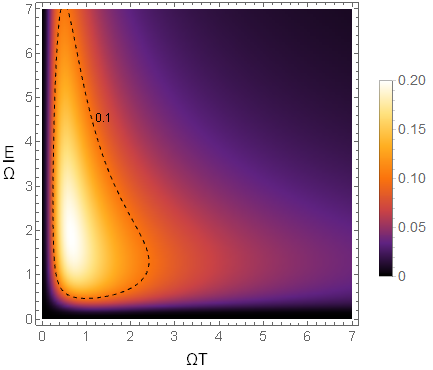}}
 \hspace{0.4cm}
\subfloat[$C_0/C_{0.8}$ ($r=0$)]{\includegraphics[scale=0.36]{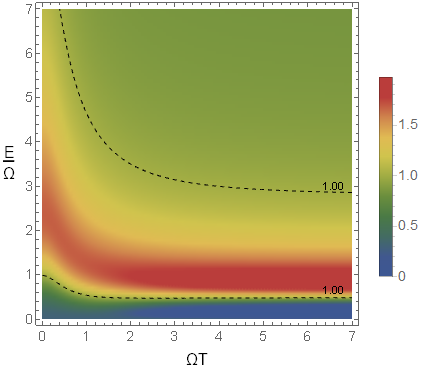}}\hspace{0.45cm}
\subfloat[$E=0.2\Omega$ ($r=0$)]{\includegraphics[scale=0.38]{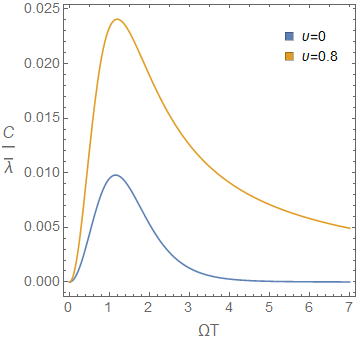}}
\end{minipage}
\caption{Left: Amount of harvested coherence, $C_{0.8}/\bar\lambda$, in $3+1$ dimensions. Center: Amount of swelling $C_0/C_{0.8}$. Right: Comparison between a static and a moving detector for an initial energy of the field $E=0.2\Omega$.}
\label{fig:7}
\end{figure*}

According to Eq. (\ref{moving-distribution}), a detector moving at a constant velocity still perceives the field as a coherent state but in a mixture of static coherent amplitude distributions of the form (\ref{amplitude}) with Doppler shifted energies equal to 
\begin{equation}
    E(\boldsymbol\upsilon)=E\gamma(1-\boldsymbol\upsilon\cdot\hat\kk).
\end{equation}
The amount of harvested coherence in this case is given by
\begin{equation}
    C_\upsilon(E,T,R)=\frac{1}{s_n}\int \frac{C(E(\boldsymbol\upsilon),T,R)}{\gamma(1-\boldsymbol\upsilon\cdot\hat\kk)^\frac{n-1}{2 }}d\hat{\kk}.
\end{equation}
In Figs. \ref{fig:6} and \ref{fig:7}  we numerically evaluate this amount for a detector moving at a constant relativistic speed of $\upsilon=0.8$, in $1+1$ and $3+1$ dimensions respectively. We observe that close to resonance the amount of coherence harvested decreases with an increasing value of the detector's speed. As in \cite{KBM}, for lower and higher initial energies of the field there exist ``swelling" regions, where it is possible to extract more coherence to a moving than to a static detector. However, this effect becomes less intense for a higher spacetime dimension.

\subsection{Assisted catalysis}
For an instantaneous interaction coherence harvesting is catalytic. Despite the fact that after each harvest the state of the field has changed, it is possible to extract the same amount of coherence to a sequence of detectors. Ignoring the trivial case of $r=0$, for a coherent amplitude distribution with phase $\phi=\frac{\pi}{2}$ each detector will harvest
\begin{equation}
    C_\upsilon(E)=\frac{2\lambda\Gamma(3/4) }{(2\pi)^\frac14}\left[\frac{E_+}{\left(1+\frac{\pi^2E_+^2}{\Omega^2}\right)^\frac34}+\frac{E_-}{\left(1+\frac{\pi^2E_-^2}{\Omega^2}\right)^\frac34}\right]
\end{equation}
units of coherence in $1+1$ and
\begin{multline}
    C_\upsilon(E)=\\\frac{16\bar\lambda \Gamma(3/4)}{(2\pi^9)^\frac14\gamma\upsilon}
    \left[\left(1+\frac{\pi^2E_-^2}{32\Omega^2}\right)^{-\frac{3}{4}}-\left(1+\frac{\pi^2E_+^2}{32\Omega^2}\right)^{-\frac{3}{4}}\right]
\end{multline}
in $3+1$ dimensions, where $E_\pm=E\gamma(1\pm\upsilon)$ denote the field's relativistic Doppler shifted energies. As has already been mentioned in Sec. \ref{sec:harvesting}, catalysis is an energy consuming process. The cost of each extraction to lowest order in this case is equal to
\begin{equation}\label{Enikost}
    \Delta E=\begin{cases}
    \frac{C_\upsilon^2(E)\Omega}{4}+\frac{\bar\lambda^2\Omega}{\pi^2}(1+\frac{\gamma}{\sqrt{2}}),&n=1\\
    &\\
    \frac{C_\upsilon^2(E)\Omega}{4}+\frac{8\bar\lambda^2\Omega}{\pi^4}\left(1+\frac{3\gamma}{\sqrt{2}}\right),&n=3.
    \end{cases}
\end{equation}
In Fig. \ref{fig:Ecost} we plot the
amount of coherence harvested through catalysis along with its energy cost (in units $\Omega$) as a function of the initial energy of the field.
For field energies close to resonance the amount obtained is maximized. Once again it can be seen that this amount decreases for an increasing value of the detector's speed. This is also true for the energy cost associated with harvesting. On the other hand, the cost associated with the vacuum remains relatively constant.
\section{Conclusions}\label{conclus}
\begin{figure*}
\subfloat[$\upsilon=0$]{\includegraphics[width=0.32\textwidth]{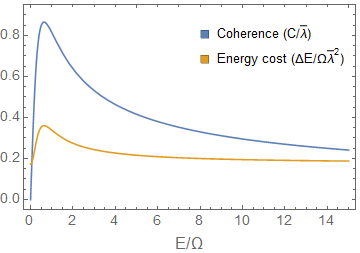}}\hspace{0.2cm}
\subfloat[$\upsilon=0.6$]{\includegraphics[width=0.32\textwidth]{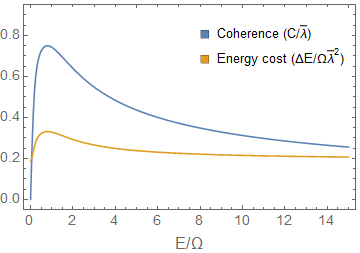}}\hspace{0.2cm}
\subfloat[$\upsilon=0.8$]{\includegraphics[width=0.32\textwidth]{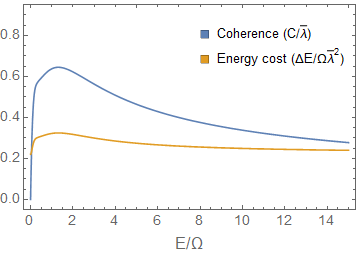}}\\
\subfloat[$\upsilon=0$]{\includegraphics[width=0.32\textwidth]{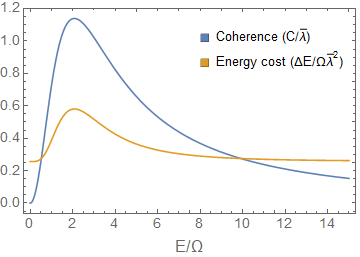}}\hspace{0.2cm}
\subfloat[$\upsilon=0.6$]{\includegraphics[width=0.32\textwidth]{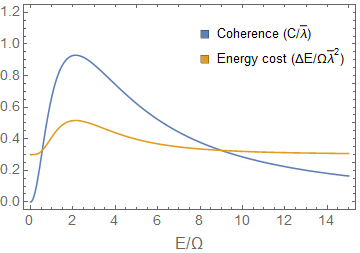}}\hspace{0.2cm}
\subfloat[$\upsilon=0.8$]{\includegraphics[width=0.32\textwidth]{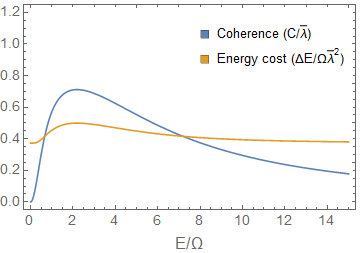}}
\caption{Amount of harvested coherence $C_\upsilon/\bar\lambda$ and cost in energy $\Delta E/\Omega\bar\lambda^2$ as a function of the initial energy of the field (in units $\Omega$) for various detector speeds. Upper: 1+1 dimensions. Lower: 3+1 dimensions.}
\label{fig:Ecost}
\end{figure*}
We have thoroughly investigated the conditions under which an UDW detector, coupled to a massless scalar field through a derivative coupling, succeeds in harvesting quantum coherence. It was proven that for an instantaneous interaction between detector and field, harvesting is catalytic, i.e., the same amount can be repeatedly extracted. For a suitable choice of the field's coherent amplitude distribution and an inertial detector, when the Fourier transform of the smearing function is positive this is also the maximum amount that can be obtained. By considering as an example a harvesting protocol in which the switching, smearing and coherent amplitude functions are Gaussian, it was demonstrated that for a coherent state of the field the process depends on the phase of the amplitude, the mean initial field energy, the mean interaction duration and the mean radius of the detector. We observed that, for a resonant energy of the field, it is possible to extend the process to longer interaction durations. It was also shown that the total time required to harvest, through repeated applications of the protocol, a single unit of coherence to a sequence of detectors is very short.

For a detector moving at a constant velocity and for a mean radius equal to the inverse of its transition frequency we verify the presence of swelling affects as was reported in \cite{KBM}. Nonetheless, since energy non-conserving interactions such as the one considered here are coherence generating \cite{PhysRevA.92.032331,BU20171670}, it is possible that this increase is due to the interaction. To avoid this possibility and in order to be able to determine how different parameters which are intrinsic to the combined system of qubit and field affect harvesting, we will study, in future work , protocols under energy conserving interactions such as the one given by the Glauber photodetection model \cite{PhysRev.130.2529,PhysRevD.46.5267} for example.
\acknowledgments{
The authors wish to thank Lena S. Peurin for fruitful discussions during preparation of this manuscript. D. M.'s research is co-financed by Greece and the  European Union (European Social Fund-ESF) through the Operational Programme ``Human Resources Development, Education and Lifelong Learning" in the context of the project ``Reinforcement of Postdoctoral Researchers - 2nd Cycle" (MIS-5033021), implemented by the State Scholarships Foundation (IKY).}

\appendix
\section{Useful relations}\label{appendix-useful}
Let
\begin{equation}
    \mathcal{F_\pm}(\kk)=\int\limits_{-\infty}^{+\infty}\chi(\tau)e^{\pm i\Omega\tau}\partial_\tau \left(F(\mathsf k,\tau)e^{i\mathsf k\cdot\mathsf x(\tau)}\right)d\tau.
\end{equation}
Taking advantage of the commutation relations between the creation and annihilation operators in Eq. (\ref{commutation}) and rewriting $\hat\Phi$ as
\begin{equation}
    \hat\Phi=\int\frac{d^n\kk}{\sqrt{(2\pi)^n2\abs{\kk}}}\left(\mathcal{F}_-(\kk)\hat a_\kk+\mathcal{F}_+^*(\kk)\hat a_\kk^\dagger\right)
\end{equation}
we can easily compute the following commutators
\begin{equation}\label{commutator}
    [\hat\Phi,\hat\Phi^\dagger]=\int\frac{d^n\kk}{{(2\pi)^n2\abs{\kk}}}\left(\abs{\mathcal{F}_-(\kk)}^2-\abs{\mathcal{F}_+(\kk)}^2\right)
\end{equation}
\begin{equation}
    [\hat\Phi,\hat H_\phi]=\int\frac{\abs{\kk}d^n\kk}{\sqrt{(2\pi)^n2\abs{\kk}}}\left(\mathcal{F}_-(\kk)\hat a_\kk-\mathcal{F}_+^*(\kk)\hat a_\kk^\dagger\right)
\end{equation}
\begin{equation}\label{comm}
    \left[[\hat\Phi,\hat H_\phi],(\hat\Phi^\dagger)^m\right]=mc^2(\hat\Phi^\dagger)^{m-1}
\end{equation}
where
\begin{equation}\label{se}
    c^2=\frac{1}{{2(2\pi)^n}}\int\left(\abs{\mathcal{F}_-(\kk)}^2+\abs{\mathcal{F}_+(\kk)}^2\right)d^n\kk.
\end{equation}
\section{Assisted catalysis for instantaneous interactions}\label{second-appendix}
For $\chi(\tau)=\delta(\tau)$ it is easy to see from Eq. (\ref{Phi}) that $\hat\Phi=\hat\Phi^\dagger$. The unitary evolution operator in Eq. (\ref{evolution}) can then be written as \cite{PhysRevD.96.065008}
\begin{equation}
\hat{U}=I\otimes \cos^{}(\lambda\hat\Phi)-i\sigma_x\otimes\sin^{}(\lambda\hat\Phi)
\end{equation}
where $\sigma_x=\ketbra{e}{g}+\ketbra{g}{e}$. Evolving the separable state of the combined system of detector and field in Eq. (\ref{separable}) and tracing out the field degrees of freedom we find that the state of the detector after the interaction is equal to
\begin{equation}\label{app:det-state}
    \rho_\text{\tiny\ensuremath D}=\left(\begin{array}{cc}
         \tr(\cos^2(\lambda\hat\Phi)\sigma_\phi)& \frac i2\tr(\sin^{}(2\lambda\hat\Phi)\sigma_\phi) \\
         -\frac i2\tr(\sin^{}(2\lambda\hat\Phi)\sigma_\phi)& \tr(\sin^2(\lambda\hat\Phi)\sigma_\phi)
    \end{array}\right).
\end{equation}
Similarly the state of the field after harvest is given by
\begin{equation}\label{app:field-state}
    \sigma'_\phi=\cos^{}(\lambda\hat\Phi)\sigma_\phi\cos^{}(\lambda\hat\Phi)+\sin^{}(\lambda\hat\Phi)\sigma_\phi\sin^{}(\lambda\hat\Phi).
\end{equation}
From Eqs. (\ref{app:det-state}) and (\ref{app:field-state}) and the definition of the $\ell_1$-norm of coherence it can be seen that the amount of harvested coherence extracted the second time is equal to
\begin{align}
    C'&=\abs{\tr(\sin^{}(2\lambda\hat\Phi)\sigma_\phi')}\nonumber\\
    &=\abs{\tr(\sin^{}(2\lambda\hat\Phi)\sigma_\phi)}
\end{align}
where in the last equality we have taken advantage of the cyclic property of the trace and the fact that $\cos^2(\lambda\hat\Phi)+\sin^2(\lambda\hat\Phi)=I_\phi$. 

We will now compute the energy difference $\Delta E$ between the initial and final states of the combined system of field plus detector and show that it is always positive. This means that catalysis is an energy consuming process so it cannot be repeated indefinitely.

From Eqs. (\ref{app:det-state}) and (\ref{app:field-state}) it is easy to see that the difference in energy before and after extraction is
\begin{align}
    \Delta E&=\tr(\hat H_\text{\tiny\ensuremath D}(\rho_\text{\tiny\ensuremath D}-\ketbra{g}))+\tr(\hat H_\phi(\sigma_\phi'-\sigma_\phi))\nonumber\\
    &=\Omega\tr(\sin^2(\lambda\hat\Phi)\sigma_\phi)\nonumber\\
    &+\frac12\tr(\left[[\cos^{}(\lambda\hat\Phi),\hat H_\phi],\cos^{}(\lambda\hat\Phi)\right]\sigma_\phi)\nonumber\\
    &+\frac12\tr(\left[[\sin^{}(\lambda\hat\Phi),\hat H_\phi],\sin^{}(\lambda\hat\Phi)\right]\sigma_\phi).
\end{align}
The first term on the right hand side as a product of two positive matrices is evidently positive, indeed this must be the case since the qubit started out in its ground state and can only gain energy. On the other hand from Eq. (\ref{comm}) it can be shown by iteration that
\begin{equation}
    \left[[\hat\Phi^\ell,\hat H_\phi],\hat\Phi^m\right]=\ell mc^2\hat\Phi^{\ell+m-2}.
\end{equation}
This means that
\begin{equation}
    \left[[\cos^{}(\lambda\hat\Phi),\hat H_\phi],\cos^{}(\lambda\hat\Phi)\right]=c^2\lambda^2\sin^2(\lambda\hat\Phi)
\end{equation}
and
\begin{equation}
    \left[[\sin^{}(\lambda\hat\Phi),\hat H_\phi],\sin^{}(\lambda\hat\Phi)\right]=c^2\lambda^2\cos^2(\lambda\hat\Phi)
\end{equation}
so finally
\begin{equation}
    \Delta E=\Omega\tr(\sin^2(\lambda\hat\Phi)\sigma_\phi)+\frac{c^2\lambda^2}{2}
\end{equation}
which is always positive.
\bibliography{harvesting}
\end{document}